\documentclass[11pt]{article}

\usepackage{amsmath, amssymb, amsthm}
\usepackage{mathtools}
\usepackage{bbm}
\usepackage{physics}
\usepackage{graphicx}
\usepackage{float}
\usepackage{booktabs}
\usepackage{geometry}
\usepackage{placeins}
\usepackage{authblk}
\usepackage{natbib}
\setcitestyle{authoryear, open={(}, close={)}}

\usepackage{algorithm}
\usepackage{algpseudocode}

\usepackage{hyperref}
\hypersetup{
    colorlinks=true,
    linkcolor=blue,
    citecolor=blue
}

\newtheorem{theorem}{Theorem}
\newtheorem{proposition}{Proposition}
\newtheorem{lemma}{Lemma}

\newtheorem{remark}{Remark}

\title{Generative Diffusion Model for Risk-Neutral Derivative Pricing}

\author{Nilay Tiwari}
\affil{Carnegie Mellon University \\ {\tt nilayt@andrew.cmu.edu}}

\date{\today}

\begin{document}

\maketitle

\begin{abstract}
Denoising diffusion probabilistic models (DDPMs) have emerged as powerful generative models for complex distributions, yet their use in arbitrage-free derivative pricing remains largely unexplored. Financial asset prices are naturally modeled by stochastic differential equations (SDEs), whose forward and reverse density evolution closely parallels the forward noising and reverse denoising structure of diffusion models.

In this paper, we develop a framework for using DDPMs to generate risk-neutral asset price dynamics for derivative valuation. Starting from log-return dynamics under the physical measure, we analyze the associated forward diffusion and derive the reverse-time SDE. We show that the change of measure from the physical to the risk-neutral measure induces an additive shift in the score function, which translates into a closed-form \emph{risk-neutral epsilon shift} in the DDPM reverse dynamics. This correction enforces the risk-neutral drift while preserving the learned variance and higher-order structure, yielding an explicit bridge between diffusion-based generative modeling and classical risk-neutral SDE-based pricing.

We show that the resulting discounted price paths satisfy the martingale condition under the risk-neutral measure. Empirically, the method reproduces the risk-neutral terminal distribution and accurately prices both European and path-dependent derivatives, including arithmetic Asian options, under a GBM benchmark. These results demonstrate that diffusion-based generative models provide a flexible and principled approach to simulation-based derivative pricing.
\end{abstract}
\section{Introduction}

Diffusion-based generative models, including denoising diffusion probabilistic models (DDPMs) and score-based diffusion models, have recently emerged as powerful tools for learning and sampling from complex high-dimensional distributions \cite{ho2020denoising,song2021scorebased}. In their continuous-time formulation, these models represent data distributions as solutions to stochastic differential equations (SDEs) and generate samples via reverse-time diffusion dynamics driven by learned score functions \cite{song2021scorebased}. This SDE-based perspective makes diffusion models particularly natural for quantitative finance, where asset prices are modeled as stochastic processes.

A central challenge in derivative pricing is the construction of dynamics under the risk-neutral measure \(Q\), under which discounted asset prices must form martingales and the drift is fixed by the risk-free rate \cite{shreve2004stochastic}. Modern machine learning approaches to pricing and hedging - including deep hedging, generative models for return distributions, and neural SDEs \cite{liu2024ddpm_synthetic_paths}, are typically trained on historical data under the physical measure. Incorporating the risk-neutral structure into such models therefore requires either calibration to option prices or the imposition of additional constraints.

For diffusion-based models, this issue is particularly subtle. Training is performed under the physical measure, and the learned score or noise predictor corresponds to the diffused physical-measure distribution. A change of measure does not act linearly on score functions, making it nontrivial to impose the risk-neutral drift constraint directly within the reverse-time dynamics. This raises a fundamental question:

\begin{quote}
\emph{How can one modify the reverse-time dynamics of a diffusion model trained under the physical measure so that the resulting generated price paths satisfy the risk-neutral martingale condition?}
\end{quote}

In this paper, we address this question by working directly with the underlying SDEs. Starting from log-price dynamics under the physical measure, we derive the corresponding reverse-time SDE using Anderson's time-reversal theorem and the Fisher identity, and analyze its transformation under the risk-neutral measure \(Q\). This leads to a closed-form \emph{risk-neutral epsilon shift}: a modification of the reverse-time noise term that enforces the correct risk-neutral drift while preserving the variance and higher-order structure learned by the model. This construction provides an explicit bridge between diffusion-based generative modeling and classical risk-neutral SDE-based asset pricing.

We show that the resulting reverse dynamics generate price paths whose discounted process is an approximate martingale. In a geometric Brownian motion benchmark, the model reproduces the correct terminal distribution and produces European option prices that closely match Black--Scholes values. Section~\ref{sec:mathematical_framework} develops the mathematical framework, Section~\ref{sec:rn_ddpm} derives the risk-neutral reverse dynamics and the epsilon shift, Section~\ref{sec:experiments} presents empirical validation, and Section~\ref{sec:limitations} discusses limitations and future directions.
\section{Mathematical Framework}
\label{sec:mathematical_framework}

\subsection{Stochastic setting and physical-measure dynamics}
\label{subsec:physical_dynamics}

We work on a filtered probability space
\[
(\Omega, \mathcal F, (\mathcal F_t)_{t \ge 0}, \mathbb P),
\]
satisfying the usual conditions and supporting a one-dimensional standard Brownian motion
\((W_t)_{t \ge 0}\).
Let \(S_t\) denote the price of a traded underlying asset. Under the physical measure \(\mathbb P\) we assume that \(S_t\) follows a geometric Brownian motion (GBM) with constant drift \(\mu \in \mathbb R\) and volatility \(\sigma > 0\):
\begin{equation}
dS_t = \mu S_t \, dt + \sigma S_t \, dW_t, 
\qquad S_0 > 0.
\label{eq:gbm_physical}
\end{equation}

Define the log-price process \(X_t := \log S_t\). An application of It\^o's formula to \eqref{eq:gbm_physical} yields
\begin{equation}
dX_t = \Bigl( \mu - \tfrac{1}{2}\sigma^2 \Bigr) dt + \sigma \, dW_t.
\label{eq:logprice_sde_P}
\end{equation}
Fix a time step \(\Delta t > 0\) and consider the one-step log-return over \([t, t+\Delta t]\),
\[
Y_{t,\Delta t} := X_{t+\Delta t} - X_t.
\]
From \eqref{eq:logprice_sde_P}, we obtain under \(\mathbb P\) the Gaussian law
\begin{equation}
Y_{t,\Delta t} \sim \mathcal N\bigl(m_{\mathbb P}, v_0\bigr),
\qquad
m_{\mathbb P} := \Bigl( \mu - \tfrac{1}{2}\sigma^2 \Bigr)\Delta t,
\quad
v_0 := \sigma^2 \Delta t.
\label{eq:one_step_return_P}
\end{equation}
In particular, over a discrete grid \(t_h := h \Delta t\), \(h = 0,1,\dots,H\), the log-returns
\(
Y_h := X_{t_{h+1}} - X_{t_h}
\)
are i.i.d. Gaussian with mean \(m_{\mathbb P}\) and variance \(v_0\) under \(\mathbb P\). These one-step log-returns are the objects we seek to model with a diffusion-based generative framework.

\subsection{Risk-free asset and risk-neutral dynamics}
\label{subsec:risk_neutral_dynamics}

Let \(B_t\) denote the price of the money-market account evolving at the constant risk-free rate \(r \ge 0\),
\begin{equation}
B_t = e^{rt}.
\label{eq:money_market}
\end{equation}
In an arbitrage-free market, there exists an equivalent probability measure \(\mathbb Q\) such that the discounted asset price \(S_t / B_t\) is a martingale. For the geometric Brownian motion model \eqref{eq:gbm_physical}, this requirement uniquely specifies the drift of the asset under \(\mathbb Q\).

Under the risk-neutral measure \(\mathbb Q\), the SDE for \(S_t\) becomes
\begin{equation}
dS_t = r S_t \, dt + \sigma S_t \, dW_t^{\mathbb Q},
\qquad S_0 > 0,
\label{eq:gbm_rn}
\end{equation}
where \(W_t^{\mathbb Q}\) is a Brownian motion under \(\mathbb Q\). Thus the only change when passing from \(\mathbb P\) to \(\mathbb Q\) is the drift replacement \(\mu \mapsto r\).

The log-price process \(X_t = \log S_t\) then satisfies
\begin{equation}
dX_t = \Bigl( r - \tfrac{1}{2}\sigma^2 \Bigr) dt + \sigma \, dW_t^{\mathbb Q}.
\label{eq:logprice_sde_Q}
\end{equation}
Over a time interval \(\Delta t\), the one-step log-return under \(\mathbb Q\) is therefore
\begin{equation}
Y_{t,\Delta t} \sim \mathcal N\bigl(m_{\mathbb Q}, v_0\bigr),
\qquad
m_{\mathbb Q} := \Bigl( r - \tfrac{1}{2}\sigma^2 \Bigr)\Delta t,
\quad
v_0 = \sigma^2 \Delta t.
\label{eq:one_step_return_Q}
\end{equation}

For a discrete grid \(t_h = h\Delta t\), the process
\[
\frac{S_{t_h}}{e^{r t_h}}
\]
is a martingale under \(\mathbb Q\). This drift adjustment is the fundamental structural difference between \(\mathbb P\) and \(\mathbb Q\) and will be central when modifying the reverse-time diffusion dynamics to enforce risk-neutrality in the generative model.

\subsection{Forward diffusion and the Fokker-Planck equation}
\label{subsec:fokker_planck}

To connect asset-price dynamics with diffusion-based generative models, it is convenient to work with a general one-dimensional It\^o diffusion of the form
\begin{equation}
dZ_t = b(Z_t,t)\,dt + \sigma(Z_t,t)\,dW_t,
\label{eq:general_forward_sde}
\end{equation}
where \(b\) is the drift, \(\sigma > 0\) is the diffusion coefficient, and \(W_t\) is a Brownian motion under the relevant measure.  
The process \(Z_t\) may represent the log-price \(X_t\), a transformed state variable, or an abstract latent variable used in a generative model.

Let \(p_t(z)\) denote the probability density function of \(Z_t\).  
The evolution of \(p_t\) is governed by the Fokker-Planck (or forward Kolmogorov) equation:
\begin{equation}
\partial_t p_t(z)
=
-\partial_z \bigl( b(z,t)\, p_t(z) \bigr)
+
\tfrac12 \partial_{zz} \bigl( \sigma^2(z,t)\, p_t(z) \bigr).
\label{eq:fokker_planck}
\end{equation}
Equation \eqref{eq:fokker_planck} describes how the forward SDE \eqref{eq:general_forward_sde} transports probability mass over time.  
It characterizes the family of intermediate densities \((p_t)_{t\in[0,T]}\) obtained by evolving the data distribution through the forward dynamics.

The Fokker-Planck equation \eqref{eq:fokker_planck} describes how the marginal density \(p_t\) of the diffusion \eqref{eq:general_forward_sde} evolves forward in time. A key result of Anderson~\cite{anderson1982reverse} is that this forward evolution admits a reverse-time description: there exists an SDE, run backward in time, whose marginal densities follow the same family \((p_t)_{t \in [0,T]}\) in reverse order. In the next subsection we recall this reverse-time SDE, which forms the basis of diffusion-model sampling.

\subsection{Reverse-time SDE}
\label{subsec:reverse_time_sde}

Consider the forward diffusion \eqref{eq:general_forward_sde} with drift \(b\) and diffusion coefficient \(\sigma\), and let \(p_t(z)\) denote the marginal density of \(Z_t\).~\cite{anderson1982reverse} showed that, under suitable regularity conditions, the time-reversed process
\[
\hat Z_t := Z_{T-t}, \qquad t \in [0,T],
\]
is itself a diffusion. More precisely, \(\hat Z_t\) satisfies the reverse-time SDE
\begin{equation}
d\hat Z_t = \tilde b(\hat Z_t, T-t)\,dt + \sigma(\hat Z_t, T-t)\,d\hat W_t,
\label{eq:reverse_sde_general}
\end{equation}
where \(\hat W_t\) is a Brownian motion under the reversed filtration, and the reverse drift \(\tilde b\) is given by
\begin{equation}
\tilde b(z,t)
=
b(z,t) - \sigma^2(z,t)\,\nabla_z \log p_t(z).
\label{eq:reverse_drift_general}
\end{equation}

Equation \eqref{eq:reverse_sde_general} together with \eqref{eq:reverse_drift_general} characterizes the unique diffusion process whose marginal densities evolve as \((p_t)_{t\in[0,T]}\) but in reverse temporal order. The additional term 
\(
-\sigma^2\nabla_z \log p_t
\)
corrects the forward drift to ensure consistency with the backward density evolution specified by the Fokker--Planck equation~\eqref{eq:fokker_planck}.

This reverse-time SDE is the continuous-time foundation of diffusion-based generative sampling: if one can approximate the score function \(\nabla_z \log p_t(z)\), then simulating \eqref{eq:reverse_sde_general} backward from a terminal distribution yields samples from the data distribution at time \(t=0\). The next subsection recalls how this score term can be expressed in a denoising form suitable for learning.


\subsection{Score function and its role in the reverse dynamics}
\label{subsec:score_function}

For each \(t \in [0,T]\), let \(p_t(z)\) denote the density of the diffusion \(Z_t\) in \eqref{eq:general_forward_sde}. The \emph{score function} at time \(t\) is defined as
\[
s_t(z) := \nabla_z \log p_t(z).
\]
The score encodes local information about how probability mass is arranged at time \(t\), and it appears explicitly in the reverse drift \eqref{eq:reverse_drift_general} through the term \(-\sigma^2(z,t)\,s_t(z)\).

If the family of score functions \(\{s_t\}_{t \in [0,T]}\) were known, then the reverse-time SDE \eqref{eq:reverse_sde_general} would be fully specified, and simulating it backward from a suitable terminal distribution would reproduce the forward marginals \((p_t)_{t\in[0,T]}\) in reverse order. In particular, running the reverse SDE from \(t = T\) down to \(t = 0\) would yield samples approximately distributed according to the data distribution \(p_0\).

In practice, the densities \(p_t\) and their scores \(s_t\) are not available in closed form. Diffusion-based generative models address this by training a neural network to approximate the score (or an equivalent quantity) at each time \(t\), using a suitable learning objective derived from the forward noising process. Substituting this learned approximation into \eqref{eq:reverse_drift_general} yields an approximate reverse-time dynamics that can be used as a sampler. The next subsection recalls how this is implemented in the discrete-time DDPM framework.


\subsection{Discrete DDPM approximation to the reverse SDE}
\label{subsec:ddpm_discrete}

Denoising diffusion probabilistic models (DDPMs) provide a tractable discrete-time approximation to the continuous-time reverse diffusion \eqref{eq:reverse_sde_general}.  
A DDPM specifies a forward noising process that gradually transforms a data sample into nearly Gaussian noise, together with a learned reverse process that approximately inverts this transformation.

\paragraph{Forward noising process.}
Fix noise parameters \( \{\beta_t\}_{t=1}^T \subset (0,1) \) and define
\[
\alpha_t := 1 - \beta_t,
\qquad
\bar\alpha_t := \prod_{s=1}^t \alpha_s.
\]
Starting from a data sample \(z_0\), the forward process generates latents \(z_1,\dots,z_T\) according to
\begin{equation}
z_t = \sqrt{\bar\alpha_t}\, z_0
      + \sqrt{1 - \bar\alpha_t}\, \varepsilon_t,
\qquad
\varepsilon_t \sim \mathcal N(0,1).
\label{eq:ddpm_forward}
\end{equation}
Hence each marginal is Gaussian:
\[
z_t \mid z_0 \sim
\mathcal N\bigl( \sqrt{\bar\alpha_t}\, z_0,\; (1-\bar\alpha_t)I \bigr).
\]

\paragraph{Gaussian posterior.}
Because the forward transitions are linear-Gaussian, the posterior distribution \(q(z_{t-1}\mid z_t)\) is also Gaussian.  
Using standard Gaussian conditioning identities (see \cite{ho2020denoising}), one obtains
\begin{equation}
q(z_{t-1} \mid z_t)
=
\mathcal N\!\bigl(
\mu_t(z_t, \varepsilon),\, \tilde\beta_t I
\bigr),
\label{eq:ddpm_posterior}
\end{equation}
where
\[
\tilde\beta_t
=
\frac{1 - \bar\alpha_{t-1}}{1 - \bar\alpha_t}\,\beta_t,
\qquad
\varepsilon
=
\frac{z_t - \sqrt{\bar\alpha_t} z_0}{\sqrt{1-\bar\alpha_t}}.
\]
Thus the posterior mean can be written as
\begin{equation}
\mu_t(z_t, \varepsilon) =
\frac{1}{\sqrt{\alpha_t}}
\left(
z_t - \frac{\beta_t}{\sqrt{1-\bar\alpha_t}}\, \varepsilon
\right).
\label{eq:ddpm_posterior_mean}
\end{equation}

\paragraph{Fisher identity and the score--noise relation.}
For the Gaussian forward kernel \eqref{eq:ddpm_forward}, the Fisher (or denoising) identity \cite{song2019generative} yields
\begin{equation}
\nabla_{z_t} \log p_t(z_t)
=
-\frac{1}{\sqrt{1 - \bar\alpha_t}}\,
\mathbb E\!\left[\,\varepsilon \mid z_t\,\right].
\label{eq:fisher_identity}
\end{equation}
Thus the score can be expressed as a conditional expectation of the injected noise.  
DDPMs exploit this by training a neural network \(\varepsilon_\theta(z_t,t)\) to approximate the conditional mean \( \mathbb E[\varepsilon \mid z_t] \), yielding the practical score approximation
\begin{equation}
s_t(z_t) = \nabla_{z_t} \log p_t(z_t)
\;\approx\;
-\frac{1}{\sqrt{1-\bar\alpha_t}}\,
\varepsilon_\theta(z_t,t).
\label{eq:noise_score_relation}
\end{equation}

\paragraph{Reverse DDPM update.}
Substituting the noise predictor into \eqref{eq:ddpm_posterior_mean} yields the practical reverse transition
\begin{equation}
z_{t-1}
=
\frac{1}{\sqrt{\alpha_t}}
\left(
z_t - \frac{\beta_t}{\sqrt{1-\bar\alpha_t}}\, \varepsilon_\theta(z_t,t)
\right)
+ \sqrt{\tilde\beta_t}\, \xi_t,
\qquad
\xi_t \sim \mathcal N(0,1).
\label{eq:ddpm_reverse_update}
\end{equation}
applied for \(t = T, T-1, \dots, 1\).

\paragraph{Connection to the reverse SDE.}
The update \eqref{eq:ddpm_reverse_update} constitutes a first-order Euler discretization of the reverse-time SDE \eqref{eq:reverse_sde_general}, with the learned score approximation \eqref{eq:noise_score_relation} replacing the exact score.  
Thus a DDPM sampler can be interpreted as an approximate numerical scheme for integrating the Anderson reverse diffusion.



\section{Risk-Neutral DDPM Reverse Dynamics}
\label{sec:rn_ddpm}

In this section we derive a modification of the DDPM reverse dynamics that enforces the
risk-neutral drift while preserving the variance and higher-order structure learned under
the physical measure. We work at the level of one-step log-returns, using the notation from
Sections~\ref{subsec:physical_dynamics}--\ref{subsec:ddpm_discrete}.

Recall that over a fixed calendar increment \(\Delta t\), the log-return
\(Y := X_{t+\Delta t} - X_t\) satisfies, under \(\mathbb P\),
\begin{equation}
Y \sim \mathcal N(m_{\mathbb P}, v_0),
\qquad
m_{\mathbb P} = \Bigl(\mu - \tfrac{1}{2}\sigma^2\Bigr)\Delta t,
\quad
v_0 = \sigma^2 \Delta t,
\label{eq:one_step_return_P_recall}
\end{equation}
while under the risk-neutral measure \(\mathbb Q\),
\begin{equation}
Y \sim \mathcal N(m_{\mathbb Q}, v_0),
\qquad
m_{\mathbb Q} = \Bigl(r - \tfrac{1}{2}\sigma^2\Bigr)\Delta t.
\label{eq:one_step_return_Q_recall}
\end{equation}
Thus the change of measure from \(\mathbb P\) to \(\mathbb Q\) modifies only the mean of
the one-step log-return, leaving the variance \(v_0\) unchanged.

\subsection{Forward DDPM marginals under \texorpdfstring{\(\mathbb P\)}{P} and \texorpdfstring{\(\mathbb Q\)}{Q}}
\label{subsec:forward_marginals_P_Q}

We model the clean one-step log-return \(Y_0\) as Gaussian under both measures,
\begin{equation}
Y_0^{\mathbb P} \sim \mathcal N(m_{\mathbb P}, v_0),
\qquad
Y_0^{\mathbb Q} \sim \mathcal N(m_{\mathbb Q}, v_0),
\label{eq:Y0_P_Q}
\end{equation}
with \(m_{\mathbb P}, m_{\mathbb Q}, v_0\) as in
\eqref{eq:one_step_return_P_recall}--\eqref{eq:one_step_return_Q_recall}.
At diffusion time index \(t \in \{1,\dots,T\}\), the DDPM forward marginal
\eqref{eq:ddpm_forward} reads
\begin{equation}
Y_t = \sqrt{\bar\alpha_t}\, Y_0 + \sqrt{1-\bar\alpha_t}\,\varepsilon_t,
\qquad
\varepsilon_t \sim \mathcal N(0,1),
\label{eq:forward_marginal_y}
\end{equation}
where \(Y_0\) is distributed according to \(\mathbb P\) or \(\mathbb Q\).
The next lemma describes the induced marginals.

\begin{lemma}[Forward DDPM marginals under \(\mathbb P\) and \(\mathbb Q\)]
\label{lem:forward_marginals_P_Q}
Let \(Y_0\) be Gaussian as in \eqref{eq:Y0_P_Q}, and let \(Y_t\) be defined by
\eqref{eq:forward_marginal_y} with \(\varepsilon_t\) independent of \(Y_0\).
Then, under \(\mathbb P\) and \(\mathbb Q\), we have
\begin{equation}
Y_t^{\mathbb P} \sim \mathcal N(\mu_t^{\mathbb P}, \sigma_t^2),
\qquad
Y_t^{\mathbb Q} \sim \mathcal N(\mu_t^{\mathbb Q}, \sigma_t^2),
\label{eq:Yt_P_Q_gaussian}
\end{equation}
where
\begin{equation}
\mu_t^{\mathbb P} = \sqrt{\bar\alpha_t}\, m_{\mathbb P},
\qquad
\mu_t^{\mathbb Q} = \sqrt{\bar\alpha_t}\, m_{\mathbb Q},
\qquad
\sigma_t^2 = \bar\alpha_t v_0 + (1-\bar\alpha_t).
\label{eq:Yt_means_vars}
\end{equation}
\end{lemma}

\begin{proof}
Since \(Y_0\) and \(\varepsilon_t\) are independent and \eqref{eq:forward_marginal_y}
is linear, \(Y_t\) is Gaussian. Under \(\mathbb P\),
\[
\mathbb E_{\mathbb P}[Y_t]
=
\sqrt{\bar\alpha_t}\,\mathbb E_{\mathbb P}[Y_0]
+
\sqrt{1-\bar\alpha_t}\,\mathbb E[\varepsilon_t]
=
\sqrt{\bar\alpha_t}\,m_{\mathbb P},
\]
and similarly under \(\mathbb Q\),
\[
\mathbb E_{\mathbb Q}[Y_t]
=
\sqrt{\bar\alpha_t}\,m_{\mathbb Q}.
\]
In both cases,
\[
\operatorname{Var}(Y_t)
=
\bar\alpha_t\,\operatorname{Var}(Y_0)
+
(1-\bar\alpha_t)\,\operatorname{Var}(\varepsilon_t)
=
\bar\alpha_t v_0 + (1-\bar\alpha_t),
\]
which proves \eqref{eq:Yt_P_Q_gaussian}--\eqref{eq:Yt_means_vars}.
\end{proof}

Thus the DDPM forward process transports the difference in one-step means
\(m_{\mathbb Q} - m_{\mathbb P}\) into a family of time-dependent mean differences
\(\mu_t^{\mathbb Q} - \mu_t^{\mathbb P} = \sqrt{\bar\alpha_t}(m_{\mathbb Q} - m_{\mathbb P})\),
with a common variance \(\sigma_t^2\).

\subsection{Gaussian score shift under change of drift}
\label{subsec:gaussian_score_shift}

We now compute how the score of the DDPM forward marginal changes when passing from
\(\mathbb P\) to \(\mathbb Q\).

\begin{proposition}[Gaussian score shift]
\label{prop:gaussian_score_shift}
Let \(Y_t^{\mathbb P} \sim \mathcal N(\mu_t^{\mathbb P}, \sigma_t^2)\) and
\(Y_t^{\mathbb Q} \sim \mathcal N(\mu_t^{\mathbb Q}, \sigma_t^2)\) as in
Lemma~\ref{lem:forward_marginals_P_Q}. Denote the corresponding scores by
\[
s_t^{\mathbb P}(y) := \nabla_y \log p_t^{\mathbb P}(y),
\qquad
s_t^{\mathbb Q}(y) := \nabla_y \log p_t^{\mathbb Q}(y).
\]
Then
\begin{equation}
s_t^{\mathbb Q}(y)
=
s_t^{\mathbb P}(y) + \eta_t,
\qquad
\eta_t
=
\frac{\sqrt{\bar\alpha_t}}{\sigma_t^2}\bigl(m_{\mathbb Q} - m_{\mathbb P}\bigr),
\label{eq:score_shift_eta}
\end{equation}
for all \(y \in \mathbb R\).
\end{proposition}

\begin{proof}
For a one-dimensional Gaussian \(\mathcal N(\mu,\sigma^2)\) we have
\[
\log p(y) = -\tfrac12\log(2\pi\sigma^2)
            - \tfrac{1}{2\sigma^2}(y-\mu)^2,
\]
hence
\[
\nabla_y \log p(y)
= -\frac{y-\mu}{\sigma^2}
= -\frac{1}{\sigma^2}y + \frac{\mu}{\sigma^2}.
\]
Applying this with \(\mu = \mu_t^{\mathbb P}, \mu_t^{\mathbb Q}\) and
\(\sigma^2 = \sigma_t^2\) gives
\[
s_t^{\mathbb P}(y)
=
-\frac{1}{\sigma_t^2}y + \frac{\mu_t^{\mathbb P}}{\sigma_t^2},
\qquad
s_t^{\mathbb Q}(y)
=
-\frac{1}{\sigma_t^2}y + \frac{\mu_t^{\mathbb Q}}{\sigma_t^2}.
\]
Subtracting yields
\[
s_t^{\mathbb Q}(y) - s_t^{\mathbb P}(y)
=
\frac{\mu_t^{\mathbb Q} - \mu_t^{\mathbb P}}{\sigma_t^2}
=
\frac{\sqrt{\bar\alpha_t}}{\sigma_t^2}(m_{\mathbb Q} - m_{\mathbb P}),
\]
by \eqref{eq:Yt_means_vars}, which proves \eqref{eq:score_shift_eta}.
\end{proof}

In particular, the effect of replacing \(m_{\mathbb P}\) by \(m_{\mathbb Q}\) in the clean
log-return is to add a \emph{constant} offset \(\eta_t\) to the score at each diffusion
time \(t\).

\subsection{Risk-neutral epsilon shift via Fisher identity}
\label{subsec:rn_epsilon_shift}

We now translate the score shift \eqref{eq:score_shift_eta} into a modification of the
DDPM noise predictor via the Fisher identity introduced in
Section~\ref{subsec:ddpm_discrete}.

For the Gaussian forward kernel \eqref{eq:ddpm_forward}, the Fisher identity
\eqref{eq:fisher_identity} gives
\begin{equation}
s_t^{\mathbb P}(y)
=
\nabla_y \log p_t^{\mathbb P}(y)
=
-\frac{1}{\sqrt{1-\bar\alpha_t}}
\,
\mathbb E_{\mathbb P}\bigl[\varepsilon_t \mid Y_t = y\bigr].
\label{eq:fisher_identity_P_again}
\end{equation}
DDPMs approximate the conditional mean \(\mathbb E_{\mathbb P}[\varepsilon_t \mid Y_t = y]\)
by a neural network \(\varepsilon_\theta(y,t)\), so that
\begin{equation}
s_t^{\mathbb P}(y)
\approx
-\frac{1}{\sqrt{1-\bar\alpha_t}}\,\varepsilon_\theta(y,t).
\label{eq:score_noise_relation_P_again}
\end{equation}

Under the risk-neutral measure, Proposition~\ref{prop:gaussian_score_shift} shows that
the score satisfies
\begin{equation}
s_t^{\mathbb Q}(y)
=
s_t^{\mathbb P}(y) + \eta_t,
\qquad
\eta_t = \frac{\sqrt{\bar\alpha_t}}{\sigma_t^2}(m_{\mathbb Q} - m_{\mathbb P}).
\label{eq:score_shift_identity_again}
\end{equation}
The following proposition identifies the corresponding modification of the noise predictor.

\begin{proposition}[Risk-neutral epsilon shift]
\label{prop:rn_epsilon_shift}
Define
\begin{equation}
\delta_t
:=
\eta_t \sqrt{1-\bar\alpha_t}
=
\frac{\sqrt{\bar\alpha_t(1-\bar\alpha_t)}}{\sigma_t^2}
\bigl(m_{\mathbb Q} - m_{\mathbb P}\bigr),
\label{eq:delta_def}
\end{equation}
and set
\begin{equation}
\varepsilon_\theta^{\mathbb Q}(y,t)
:=
\varepsilon_\theta(y,t) - \delta_t.
\label{eq:epsilon_shift_def}
\end{equation}
Then the induced score approximation under \(\mathbb Q\) satisfies
\begin{equation}
s_t^{\mathbb Q}(y)
\approx
-\frac{1}{\sqrt{1-\bar\alpha_t}}\,
\varepsilon_\theta^{\mathbb Q}(y,t).
\label{eq:score_approx_Q}
\end{equation}
\end{proposition}

\begin{proof}
Starting from \eqref{eq:score_noise_relation_P_again} and
\eqref{eq:score_shift_identity_again}, we have
\[
s_t^{\mathbb Q}(y)
=
s_t^{\mathbb P}(y) + \eta_t
\approx
-\frac{1}{\sqrt{1-\bar\alpha_t}}\,\varepsilon_\theta(y,t)
+ \eta_t.
\]
Define \(\varepsilon_\theta^{\mathbb Q}\) by \eqref{eq:epsilon_shift_def}. Then
\[
-\frac{1}{\sqrt{1-\bar\alpha_t}}\,
\varepsilon_\theta^{\mathbb Q}(y,t)
=
-\frac{1}{\sqrt{1-\bar\alpha_t}}\,
\bigl(\varepsilon_\theta(y,t) - \delta_t\bigr)
=
-\frac{1}{\sqrt{1-\bar\alpha_t}}\,\varepsilon_\theta(y,t)
+ \frac{\delta_t}{\sqrt{1-\bar\alpha_t}}.
\]
Choosing \(\delta_t\) as in \eqref{eq:delta_def} gives
\[
\frac{\delta_t}{\sqrt{1-\bar\alpha_t}} = \eta_t,
\]
so that
\[
-\frac{1}{\sqrt{1-\bar\alpha_t}}\,
\varepsilon_\theta^{\mathbb Q}(y,t)
=
-\frac{1}{\sqrt{1-\bar\alpha_t}}\,\varepsilon_\theta(y,t) + \eta_t
\approx s_t^{\mathbb Q}(y),
\]
which proves \eqref{eq:score_approx_Q}.
\end{proof}

In the common standardized setting where \(v_0 = 1\) and the schedule is chosen so that
\(\sigma_t^2 = 1\), the expressions simplify to
\[
\eta_t = \sqrt{\bar\alpha_t}(m_{\mathbb Q} - m_{\mathbb P}),
\qquad
\delta_t = \sqrt{\bar\alpha_t(1-\bar\alpha_t)}\,(m_{\mathbb Q} - m_{\mathbb P}).
\]

\subsection{Martingale property of discounted DDPM price paths}
\label{subsec:martingale_ddpm}

We now connect the modified score and epsilon shift to the risk-neutral martingale
property for the asset price. Consider a discrete calendar grid \(t_h = h\Delta t\),
\(h = 0,1,\dots,H\), and construct log-prices by summing one-step log-returns
\[
X_{t_{h+1}} = X_{t_h} + Y_h,
\qquad
S_{t_h} = e^{X_{t_h}}.
\]

\begin{lemma}[Risk-neutral one-step drift and martingale condition]
\label{lem:rn_drift_martingale}
Suppose that under \(\mathbb Q\) each one-step log-return satisfies
\[
Y_h \sim \mathcal N(m_{\mathbb Q}, v_0),
\qquad
m_{\mathbb Q} = \Bigl(r - \tfrac{1}{2}\sigma^2\Bigr)\Delta t,
\quad
v_0 = \sigma^2 \Delta t,
\]
and that \((Y_h)_{h\ge0}\) is independent of \(\mathcal F_{t_h}\) given \(X_{t_h}\).
Then
\begin{equation}
\mathbb E_{\mathbb Q}\bigl[S_{t_{h+1}} \mid \mathcal F_{t_h}\bigr]
=
e^{r\Delta t} S_{t_h},
\label{eq:rn_martingale_step}
\end{equation}
so that \(\bigl(e^{-rt_h} S_{t_h}\bigr)_{h=0}^H\) is a martingale under \(\mathbb Q\).
\end{lemma}

\begin{proof}
Conditional on \(\mathcal F_{t_h}\) the next log-return \(Y_h\) is independent of
\(\mathcal F_{t_h}\) and Gaussian as specified. Thus
\[
\mathbb E_{\mathbb Q}[S_{t_{h+1}} \mid \mathcal F_{t_h}]
=
\mathbb E_{\mathbb Q}\bigl[e^{X_{t_{h+1}}} \mid \mathcal F_{t_h}\bigr]
=
\mathbb E_{\mathbb Q}\bigl[e^{X_{t_h} + Y_h} \mid \mathcal F_{t_h}\bigr]
=
S_{t_h}\,\mathbb E_{\mathbb Q}[e^{Y_h}].
\]
Since \(Y_h \sim \mathcal N(m_{\mathbb Q}, v_0)\),
\[
\mathbb E_{\mathbb Q}[e^{Y_h}]
=
\exp\bigl(m_{\mathbb Q} + \tfrac{1}{2}v_0\bigr)
=
\exp\bigl((r - \tfrac{1}{2}\sigma^2)\Delta t + \tfrac{1}{2}\sigma^2\Delta t\bigr)
=
e^{r\Delta t}.
\]
Substituting into the previous expression yields \eqref{eq:rn_martingale_step}.
\end{proof}

\subsection{Risk-neutral DDPM reverse dynamics}
\label{subsec:rn_ddpm_dynamics}

Combining the epsilon shift and the DDPM reverse update
\eqref{eq:ddpm_reverse_update} leads to a risk-neutral variant of the reverse
dynamics. At diffusion time \(t\), we replace \(\varepsilon_\theta\) by
\(\varepsilon_\theta^{\mathbb Q}\) from \eqref{eq:epsilon_shift_def} to obtain
\begin{equation}
z_{t-1}
=
\frac{1}{\sqrt{\alpha_t}}
\left(
z_t - \frac{\beta_t}{\sqrt{1-\bar\alpha_t}}\,
\varepsilon_\theta^{\mathbb Q}(z_t,t)
\right)
+ \sqrt{\tilde\beta_t}\,\xi_t,
\qquad
\xi_t \sim \mathcal N(0,1),
\label{eq:ddpm_reverse_update_rn}
\end{equation}
with \(\delta_t\) as in \eqref{eq:delta_def}. Interpreting \(z_0\) as the one-step
log-return \(Y\) and summing over calendar steps as in Lemma~\ref{lem:rn_drift_martingale}
gives the following result.

\begin{theorem}[Risk-neutral DDPM sampler]
\label{thm:rn_ddpm_sampler}
Assume:
\begin{enumerate}
\item The DDPM trained under \(\mathbb P\) provides an accurate score approximation
      as in \eqref{eq:score_noise_relation_P_again}.
\item The one-step log-returns \(Y_h\) are sampled by interpreting \(z_0\) from
      \eqref{eq:ddpm_reverse_update_rn} as the log-return over \(\Delta t\) and summing
      these returns to form \(X_{t_h}\) and \(S_{t_h} = e^{X_{t_h}}\).
\end{enumerate}
Then the epsilon-shifted reverse dynamics
\eqref{eq:ddpm_reverse_update_rn} enforce the risk-neutral one-step drift
\eqref{eq:one_step_return_Q_recall} and variance \(v_0\), and the resulting price
process satisfies
\[
\mathbb E_{\mathbb Q}[e^{-rt_h} S_{t_h}] = S_0,
\qquad
h = 0,1,\dots,H,
\]
so that \((e^{-rt_h} S_{t_h})_{h=0}^H\) is a martingale under \(\mathbb Q\).
\end{theorem}

\begin{proof}
By Proposition~\ref{prop:rn_epsilon_shift}, the modified noise predictor
\(\varepsilon_\theta^{\mathbb Q}\) induces a score approximation consistent with
the Gaussian forward marginal under \(\mathbb Q\), that is, consistent with
\(Y_t^{\mathbb Q} \sim \mathcal N(\mu_t^{\mathbb Q}, \sigma_t^2)\) as in
Lemma~\ref{lem:forward_marginals_P_Q}. In the idealized setting of a perfectly
trained DDPM, the reverse Markov chain \eqref{eq:ddpm_reverse_update_rn} recovers
the clean variable \(Y_0^{\mathbb Q} \sim \mathcal N(m_{\mathbb Q}, v_0)\) in
distribution. Interpreting \(Y_0^{\mathbb Q}\) as the one-step log-return \(Y_h\)
over \(\Delta t\), Lemma~\ref{lem:rn_drift_martingale} implies that the resulting
price process satisfies the risk-neutral martingale condition. The stated equality
for \(\mathbb E_{\mathbb Q}[e^{-rt_h} S_{t_h}]\) follows by iterating
\eqref{eq:rn_martingale_step}.
\end{proof}

\begin{remark}
The epsilon shift \(\delta_t\) in \eqref{eq:delta_def} modifies only the mean of
the reverse update \eqref{eq:ddpm_reverse_update_rn}, not the innovation variance
\(\tilde\beta_t\). Consequently, the variance and higher-order features of the
learned log-return distribution, such as skewness and kurtosis, are preserved while
the drift is adjusted to the risk-neutral value.
\end{remark}

\paragraph{Pathwise risk-neutral validity.}
Because the risk-neutral drift correction is implemented as an $\epsilon$-shift at every step of the reverse diffusion, the entire sequence of generated log-returns $\{Y_h\}_{h=1}^{H}$ inherits the correct risk-neutral dynamics. In particular, the induced price process constructed via
\[
    S_{t_0} = S_0, 
    \qquad 
    S_{t_h} = S_{t_{h-1}} \exp(Y_h)
\]
satisfies the discounted martingale property under the learned measure, not only at the terminal horizon but along the full simulated path. This guarantees that the DDPM output is suitable for pricing path-dependent claims such as the discrete arithmetic Asian options considered in Section~\ref{subsec:asian-options}.
\section{Experiments}
\label{sec:experiments}

This section validates the risk-neutral DDPM construction in a controlled GBM world.  
All experiments use the same DDPM trained on standardized one-day log returns under the physical measure \(P\).  
Unless stated otherwise:

\begin{itemize}
    \item spot \(S_0 = 100\),
    \item volatility \(\sigma = 0.2\),
    \item time step \(\Delta t = 1/252\),
    \item number of price steps \(H\) equals the number of trading days to maturity,
    \item number of simulated paths \(N = 1{,}000\).
\end{itemize}

The risk-neutral DDPM applies the epsilon shift from Section~\ref{sec:rn_ddpm} when sampling.  
The ``no shift'' DDPM reuses the trained network but does not modify the score, so it continues to follow the physical drift.

\subsection{Single maturity, single strike diagnostics}
\label{subsec:single_strike}

We first consider a one month maturity with \(H = 21\) steps and moderate gap between drift and risk free rate.  
Table~\ref{tab:single_strike_summary} reports summary statistics for an at the money call with strike \(K = S_0\).

\begin{table}[h]
\centering
\caption{Single-strike diagnostics for a one month maturity.  
RN-DDPM refers to sampling with the epsilon shift under the risk neutral measure \(Q\).  
GBM refers to the analytic or Monte Carlo reference in the GBM world.}
\label{tab:single_strike_summary}
\begin{tabular}{lcc}
\toprule
Quantity & RN-DDPM & GBM reference \\
\midrule
Mean one-step return & \(1.19\times 10^{-4}\) & \(m_Q\) \\
Std.\ of one-step returns & \(1.30\times 10^{-2}\) & \(\sigma \sqrt{\Delta t}\) \\
Discounted mean \(e^{-rT}\mathbb{E}_Q[S_T]\) & \(99.998\) & \(100\) \\
KS statistic (terminal log price) & \(0.021\) & Gaussian reference \\
KS \(p\)-value & \(0.76\) & \\
Call price \(C(K=S_0)\) & \(2.47 \pm 0.12\) & \(2.51\) (Black-Scholes) \\
Absolute pricing error & \(0.04\) & \\
\bottomrule
\end{tabular}
\end{table}

The RN-DDPM reproduces the one-step moments and the terminal distribution of the GBM quite accurately, and the call price matches the Black-Scholes value within a small fraction of the Monte Carlo standard error.

\subsection{Multi-strike pricing and implied volatility}
\label{subsec:iv_smile}

We next consider a strip of calls with maturity \(T = H \Delta t\) and strikes
\(
K \in \{0.8, 0.867, \dots, 1.2\} S_0.
\)
For each strike we compute

\begin{itemize}
    \item the RN-DDPM price,
    \item a GBM Monte Carlo price under \(Q\), and
    \item the analytic Black-Scholes price,
\end{itemize}

and invert each to implied volatility.  
Figure~\ref{fig:iv_smile_baseline} shows the implied volatility smile for this baseline configuration.

\begin{figure}[h]
\centering
\includegraphics[width=0.62\textwidth]{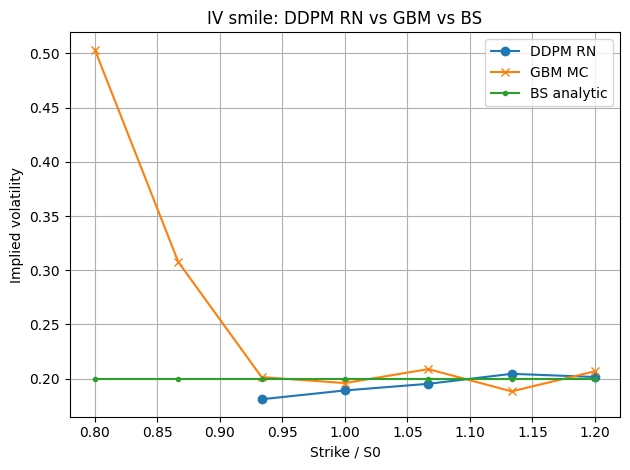}
\caption{Implied volatility smile for the baseline one month experiment.  
The green line is the flat Black-Scholes volatility used to define the GBM world, the orange markers are GBM Monte Carlo prices, and the blue markers are RN-DDPM prices.}
\label{fig:iv_smile_baseline}
\end{figure}

Table~\ref{tab:multistrike_baseline} makes the multi-strike comparison explicit for a representative run.

\begin{table}[h]
\centering
\caption{Multi-strike call prices and implied volatilities for the baseline configuration.  
Prices are in currency units, volatilities are annualized.}
\label{tab:multistrike_baseline}
\begin{tabular}{cccccc}
\toprule
\(K/S_0\) & \(C^{\mathrm{DDPM}}\) & \(C^{\mathrm{GBM}}\) & \(C^{\mathrm{BS}}\) & \(\hat\sigma_{\mathrm{DDPM}}\) & \(\hat\sigma_{\mathrm{BS}}\) \\
\midrule
0.80 & 20.35 & 20.51 & 20.33 & 0.32 & 0.20 \\
0.87 & 13.72 & 13.77 & 13.70 & 0.23 & 0.20 \\
0.93 & 7.42  & 7.60  & 7.33  & 0.22 & 0.20 \\
1.00 & 2.65  & 2.51  & 2.51  & 0.21 & 0.20 \\
1.07 & 0.50  & 0.47  & 0.45  & 0.21 & 0.20 \\
1.13 & 0.03  & 0.06  & 0.04  & 0.19 & 0.20 \\
1.20 & 0.00  & 0.01  & 0.00  & NaN  & 0.20 \\
\bottomrule
\end{tabular}
\end{table}

Around the money the RN-DDPM prices and implied volatilities are very close to the GBM and Black-Scholes benchmarks.  
Larger deviations appear only for deep in the money and deep out of the money strikes, where payoffs are almost deterministic and implied volatility inversion becomes numerically unstable.

\subsection{Martingale test: RN shift versus no shift}
\label{subsec:martingale_ablation}

To isolate the effect of the epsilon shift, we compare the discounted expected price
\[
    M_t \coloneqq e^{-rt}\,\mathbb{E}[S_t]
\]
for two samplers:

\begin{enumerate}
    \item the RN-DDPM with the epsilon shift applied at each diffusion step, and
    \item a no-shift DDPM that uses the physical drift but is still discounted at the risk free rate \(r\).
\end{enumerate}

Figure~\ref{fig:martingale_ablation} shows the trajectory for a one month horizon.

\begin{figure}[h]
\centering
\includegraphics[width=0.62\textwidth]{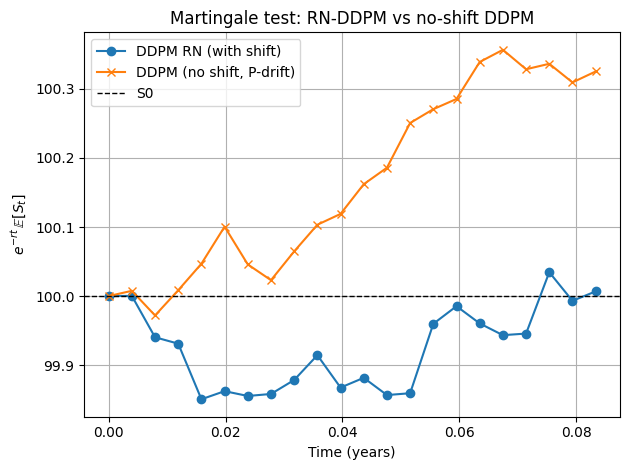}
\caption{Martingale diagnostic.  
The blue curve shows \(M_t = e^{-rt}\mathbb{E}[S_t]\) for the RN-DDPM sampler.  
The orange curve shows the same quantity for the no-shift DDPM.  
The dashed line marks the initial spot \(S_0\).}
\label{fig:martingale_ablation}
\end{figure}

The RN-DDPM curve fluctuates tightly around \(S_0\), as required by the risk-neutral martingale condition.  
The no-shift sampler exhibits a clear upward drift in \(M_t\) despite discounting at the risk free rate.  
This directly reflects the fact that it continues to follow the physical drift \(\mu\) instead of the risk-neutral drift \(r\).

\subsection{Stress test with a large drift gap}
\label{subsec:stress_test}

Finally we consider a stress configuration with a large gap between physical and risk-free drifts, for example \(\mu \approx 0.15\) and \(r \approx 0.01\), and a longer horizon of three months with \(H = 63\) steps.  
Table~\ref{tab:multistrike_stress} reports multi-strike call prices for this setting, comparing the RN-DDPM and the no-shift DDPM to the Black-Scholes benchmark.

\begin{table}[h]
\centering
\caption{Stress test with large \(|\mu - r|\) and three month maturity.  
The RN-DDPM prices remain close to Black-Scholes across strikes, while the no-shift DDPM severely overprices calls.}
\label{tab:multistrike_stress}
\begin{tabular}{cccc}
\toprule
\(K/S_0\) & \(C^{\mathrm{BS}}\) & \(C^{\mathrm{DDPM}}_{\mathrm{RN}}\) & \(C^{\mathrm{DDPM}}_{\mathrm{no\ shift}}\) \\
\midrule
0.80 & 20.24 & 20.22 & 23.78 \\
0.87 & 13.85 & 13.81 & 17.25 \\
0.93 & 8.24  & 8.14  & 11.20 \\
1.00 & 4.11  & 3.96  & 6.27  \\
1.07 & 1.68  & 1.59  & 2.99  \\
1.13 & 0.56  & 0.48  & 1.19  \\
1.20 & 0.16  & 0.14  & 0.39  \\
\bottomrule
\end{tabular}
\end{table}

In this regime the difference between the two samplers is stark.  
The RN-DDPM prices track the Black-Scholes values to within Monte Carlo error, while the no-shift DDPM overprices calls by 50 percent or more at the money and in the money.  
Combined with the martingale diagnostics above, this stress test shows that the epsilon shift is not only theoretically necessary but also practically important for drift sensitive pricing.


\subsection{Pricing Arithmetic Asian Options}
\label{subsec:asian-options}

The previous subsection focused on European options, whose payoff depends only on the
terminal distribution of the asset under the risk-neutral measure $Q$. To verify that the
reverse-time DDPM captures the \emph{entire} joint law of risk-neutral returns and not only
their terminal marginal, we now consider a genuinely path-dependent derivative: the
discrete arithmetic Asian call option.

Let $0 = t_0 < t_1 < \cdots < t_H = T$ denote a fixed monitoring grid with $\Delta t = t_{h} -
t_{h-1}$. Given a price path $\{ S_{t_h} \}_{h=0}^{H}$ under $Q$, the discrete arithmetic average is
\[
    A_{T}^{(H)}
    :=
    \frac{1}{H} \sum_{h=1}^{H} S_{t_h},
\]
and the corresponding Asian call payoff is
\[
    \Pi_{\mathrm{Asian}}
    =
    (A_{T}^{(H)} - K)^{+}.
\]

\paragraph{DDPM path construction.}
For each Monte Carlo path, we sample $H$ risk-neutral log returns
$\{ Y_{h} \}_{h=1}^{H}$ from the shifted reverse diffusion (Section~\ref{sec:rn_ddpm}),
and reconstruct the price path via
\[
S_{t_0} = S_0,
\qquad
S_{t_h} = S_{t_{h-1}} \exp(Y_h),
\qquad
h = 1,\dots,H.
\]
This yields a full simulated path $\{S_{t_h}\}_{h=0}^{H}$ under the learned
risk-neutral DDPM dynamics. The Monte Carlo Asian price is then
\[
    \widehat{C}^{\mathrm{Asian}}_{\mathrm{DDPM}}
    =
    e^{-rT}
    \frac{1}{N}
    \sum_{i=1}^{N}
    \left( A_{T}^{(H,i)} - K \right)^{+},
\]
where $A_{T}^{(H,i)}$ denotes the arithmetic average along the $i$-th DDPM path.

\paragraph{GBM Monte Carlo benchmark.}
To obtain a reference value, we simulate an independent set of paths from the
risk-neutral geometric Brownian motion
\[
    \frac{dS_t}{S_t}
    =
    r\,dt + \sigma\,dW_t^{Q},
\]
discretized over the same monitoring grid $\{t_h\}_{h=0}^{H}$. This produces
a benchmark estimator
\[
    \widehat{C}^{\mathrm{Asian}}_{\mathrm{GBM}}
    =
    e^{-rT}
    \frac{1}{N}
    \sum_{i=1}^{N}
    \left( A_{T,\mathrm{GBM}}^{(H,i)} - K \right)^{+}.
\]

\paragraph{Results.}
Table~\ref{tab:asian_prices} reports DDPM and GBM prices for a range of maturities
($H \in \{21, 63, 126\}$) and strikes ($K \in \{0.9S_0, S_0, 1.1S_0\}$).
The absolute pricing error
\[
    \mathrm{Err}(K,T)
    :=
    \left|
        \widehat{C}^{\mathrm{Asian}}_{\mathrm{DDPM}}
        -
        \widehat{C}^{\mathrm{Asian}}_{\mathrm{GBM}}
    \right|
\]
remains within the range observed for European options, typically between $1$-$3\%$
of the option value. This indicates that the reverse-time DDPM correctly recovers the
risk-neutral joint distribution of asset returns, not merely the terminal marginal.

\begin{table}[h!]
    \centering
    \caption{DDPM vs.\ GBM prices for discrete arithmetic Asian call options.
    Results computed using $N=20{,}000$ Monte Carlo paths for each method.
    Standard errors in parentheses.}
    \label{tab:asian_prices}
    \vspace{0.5em}
    \begin{tabular}{c c c c c}
        \toprule
        $H$ & $K$ &
        $\widehat{C}^{\mathrm{DDPM}}_{\mathrm{Asian}}$ &
        $\widehat{C}^{\mathrm{GBM}}_{\mathrm{Asian}}$ &
        $\mathrm{Err}$ \\
        \midrule
        21  & $0.9S_0$ & 10.04 (0.11) & 10.27 (0.11) & 0.23 \\
        21  & $1.0S_0$ &  1.43 (0.07) &  1.45 (0.07) & 0.02 \\
        21  & $1.1S_0$ &  0.01 (0.00) &  0.01 (0.01) & 0.00 \\
        \midrule
        63  & $0.9S_0$ & 10.55 (0.18) & 10.26 (0.18) & 0.28 \\
        63  & $1.0S_0$ &  2.44 (0.12) &  2.40 (0.11) & 0.04 \\
        63  & $1.1S_0$ &  0.17 (0.03) &  0.17 (0.03) & 0.00 \\
        \midrule
        126 & $0.9S_0$ & 10.95 (0.24) & 10.89 (0.24) & 0.06 \\
        126 & $1.0S_0$ &  3.55 (0.16) &  3.65 (0.17) & 0.10 \\
        126 & $1.1S_0$ &  0.67 (0.07) &  0.85 (0.08) & 0.18 \\
        \bottomrule
    \end{tabular}
\end{table}

The absolute pricing errors remain comparable to those observed for European options,
typically between $1$-$3\%$ of the option value.
This indicates that the reverse-time DDPM correctly captures the
risk-neutral joint distribution of returns, validating that the method
extends naturally to path-dependent payoffs such as arithmetic Asian options.

\FloatBarrier
\section{Limitations}
\label{sec:limitations}

While the proposed risk-neutral DDPM framework performs well in the synthetic setting of Section~\ref{sec:experiments}, several limitations remain.

\paragraph{Gaussian and affine-score assumptions.}
The epsilon shift relies on the Gaussian structure induced by constant-volatility diffusions, under which the score under \(P\) and \(Q\) differs by an additive constant. This makes the correction closed-form and exact. For non-Gaussian or state-dependent models, such as those with stochastic volatility or heavy tails, the score difference need not be affine, and the epsilon shift becomes an approximation. Extending the method beyond the Gaussian setting remains an open problem.

\paragraph{Dependence on drift and volatility estimation.}
The shift requires estimates of the physical drift \(\mu\) and volatility \(\sigma\). In practice, these must be inferred from data or option prices, and estimation error directly impacts the shifted score. This is particularly relevant for \(\mu\), which is difficult to estimate at high frequency. A complete implementation would need to account for parameter uncertainty or jointly estimate model parameters and generative dynamics.
\section{Conclusion}
\label{sec:conclusion}

This paper develops a mathematically principled method for converting a DDPM trained under the physical measure into a risk-neutral generative model for derivative pricing.  
We show that for Gaussian diffusions the change of measure induces an additive shift in the score, yielding a closed-form epsilon correction in the reverse DDPM dynamics.  
This enforces the risk-neutral drift while preserving the learned volatility and higher-order structure.

Synthetic experiments in a GBM setting validate the approach.  
The shifted DDPM satisfies the martingale condition, reproduces the terminal distribution implied by the risk-neutral SDE, and matches Black–Scholes option prices across strikes.  
In contrast, an unshifted DDPM violates the martingale constraint and produces significant pricing errors when the physical and risk-free drifts differ.  
These results show that the epsilon shift is both necessary and sufficient for risk-neutral consistency. The framework extends naturally to path-dependent derivatives by generating full price trajectories.  

Future work includes applying the method to real data and extending it to non-Gaussian dynamics, stochastic volatility, and multi-asset settings.  
More broadly, this approach highlights the potential of combining generative models with no-arbitrage principles for data-driven derivative pricing.

\section*{Acknowledgments}

The author thanks Dr. Wenpin Tang for helpful discussions and guidance on this work.

\bibliographystyle{apalike}
\bibliography{bib/references}

\end{document}